\documentclass[11pt, letterpaper,reqno]{amsart}

\usepackage{amssymb}
\usepackage{graphicx}
\usepackage[hyphens]{url}
\usepackage{hyperref}
\usepackage{appendix}
\usepackage{dsfont}

\numberwithin{equation}{section}

\newtheorem{theorem}{Theorem}[section]
\newtheorem{assumption}{Assumption}[section]

\newtheorem{lemma}{Lemma}[section]

\theoremstyle{definition}
\newtheorem{definition}{Definition}[section]
\theoremstyle{remark}
\newtheorem{remark}{Remark}[section]
\newtheorem{example}{Example}[section]

\newcommand{\leref}{Lemma~\ref}

\newcommand{\thref}{Theorem~\ref}
\newcommand{\reref}{Remark~\ref}
\newcommand{\asref}{Assumption~\ref}

\renewcommand{\P}{\mathbb{P}}
\newcommand{\Q}{\mathbb{Q}}
\newcommand{\R}{\mathbb{R}}
\newcommand{\E}{\mathbb{E}}

\newcommand{\cF}{\mathcal{F}}

\newcommand{\cM}{\mathcal{M}}

\newcommand{\cT}{\mathcal{T}}
\newcommand{\T}{\mathbb{T}}

\newcommand{\cX}{\mathcal{X}}
\newcommand{\cY}{\mathcal{Y}}

\newcommand{\cP}{\mathcal{P}}
\newcommand{\cQ}{\mathcal{Q}}
\newcommand{\A}{\mathcal{A}}

\newcommand{\eps}{\varepsilon}

\newcommand{\f}{\mathfrak{f}}

\newcommand{\EQUIV}{\Longleftrightarrow}
\newcommand{\cH}{\mathcal{H}}
\newcommand{\PPhi}{\Phi(H,a)}
\newcommand{\PPPhi}{\Phi_{\bar g,\bar h}(H,a,b,c,\mu)}
\newcommand{\cC}{\mathcal{C}}
\newcommand{\cD}{\mathcal{D}}
\newcommand{\bL}{\mathbb{L}}
\newcommand{\cI}{\mathcal{I}}
\newcommand{\cJ}{\mathcal{J}}
\newcommand{\cR}{\mathcal{R}}
\newcommand{\kg}{\mathfrak{g}}

\title[]{Arbitrage, hedging and utility maximization using semi-static trading strategies with American options}
\author[]{Erhan Bayraktar} \thanks{This research is supported in part by the National Science Foundation under grant DMS-0955463.}  
\address{Department of Mathematics, University of Michigan}
\email{erhan@umich.edu}
\author[]{Zhou Zhou}
\address{IMA, University of Minnesota}
\email{zhoux528@umn.edu}
\date{\today}
\begin{document}
\maketitle

\begin{abstract}
We consider a financial market where stocks are available for dynamic trading, and European and American options are available for static trading (semi-static trading strategies). We assume that the American options are infinitely divisible, and can only be bought but not sold. In the first part of the paper, we work within the framework without model ambiguity. We first get the fundamental theorem of asset pricing (FTAP). Using the FTAP, we get the dualities for the hedging prices of European and American options. Based on the hedging dualities, we also get the duality for the utility maximization. In the second part of the paper, we consider the market which admits non-dominated model uncertainty. We first establish the hedging result, and then using the hedging duality we further get the FTAP. Due to the technical difficulty stemming from the non-dominancy of the probability measure set, we use a discretization technique and apply the minimax theorem.
\end{abstract}


\section{Introduction}

The arbitrage, hedging, and utility maximization problems have been extensively studied in the field of financial mathematics. We refer to \cite{EQF,Sch2} and the references therein. Recently, there has been a lot of work on these three topics where stocks are traded dynamically and (European-style) options are traded statically (semi-static strategies, see e.g., \cite{Hobson1}). For example, \cite{Sch3,Mathias,nutz2,Hobson1} analyze the arbitrage and/or super-hedging in the setup of model free or model uncertainty, \cite{Sircar} investigates optimal hedging of barrier options under a given model, and \cite{Siorpaes} studies the utility maximization within a given model. It is worth noting that most of the literature related to semi-static strategies only consider European-style options as to be liquid options, and there are only a few papers incorporating American-style options for static trading. In particular, \cite{Campi} studies the completeness (in some $\bL^2$ sense) of the market where American put options of all the strike prices are available for semi-static trading, and \cite{Cox} studies the no arbitrage conditions on the price function of American put options where European and American put options are available.  

In this paper, we consider a financial market in discrete time consisting of stocks, (path-dependent) European options, and (path-dependent) American options (we also refer to these as hedging options), where the stocks are traded dynamically and European and American options are traded statically. We assume that the American options are infinitely divisible (i.e., we can break each unit American option into pieces, and exercise each piece separately), and we can only buy but not sell American options. 

In the first part of this paper, we consider the market without model ambiguity. We obtain the fundamental theorem of asset pricing (FTAP) under the notion of strict no arbitrage that is slightly stronger than no arbitrage in the usual sense. Then by the FTAP result, we further get dualities of the sub-hedging prices of European and American options. Using the duality result, we then study the utility maximization problem and get the duality of the value function.

In the second part of the paper, we work within the framework of model uncertainty. We use  the minimax theorem to prove the sub-hedging results for European and American options. From the sub-hedging dualities, we further get the FTAP result. Due to either the non-dominancy of the set of martingale measures, or the discontinuity of the payoff of the American options, we can not directly apply the minimax theorem in some steps of the proof. To overcome these technical difficulties, we first discretize the path space, and then apply the minimax theorem within the discretized space, and finally take a limit. A key assumption in this part is the weak compactness of some set of martingale measures with distribution constraints. This assumption is satisfied if we consider all the physical measures on a compact space, or if the liquid European options can compactify the set of martingale measures (see e.g., \cite{Sch3,Mathias}).

It is crucial to assume the infinite divisibility of the American options just like the stocks and European options. From a financial point of view, it is often the case that we can do strictly better when we break one unit of the American options into pieces and exercise each piece separately. In Section 2, we provide a motivating example in which without the divisibility assumption of the American option the no arbitrage condition holds yet there is no equivalent martingale measure (EMM) that prices the hedging options correctly. Moreover, we see in this example that the super-hedging price of the European option is not equal to the supremum of the expectation over all the EMMs which price the hedging options correctly. Mathematically, the infinite divisibility leads to the convexity and closedness of some related sets of random variables, which enables us either to apply the separating hyperplane argument to obtain the existence of an EMM that prices the options correctly in the case without model ambiguity, or to apply minimax theorem to get the sub-hedging duality in the case of model uncertainty.

The rest of the paper is organized as follows. In the next section, we will provide a motivating example. In Section 3, we shall introduce the setup and the main results of FTAP, sub-hedging, and utility maximization duality when there is no model ambiguity. We give the proof of these results in Section 4. In Section 5, we state the FTAP and sub-hedging results when the market admits non-dominated model uncertainty. Finally in Section 6, we demonstrate the proof of the model uncertainty version of the FTAP and sub-hedging results.

\section{A motivating example}
In this section, we shall look at an example of super-hedging of a European option using the stock and the American option. This example will motivate us to consider the divisibility of American options.
\begin{center}
\includegraphics[scale=0.17]{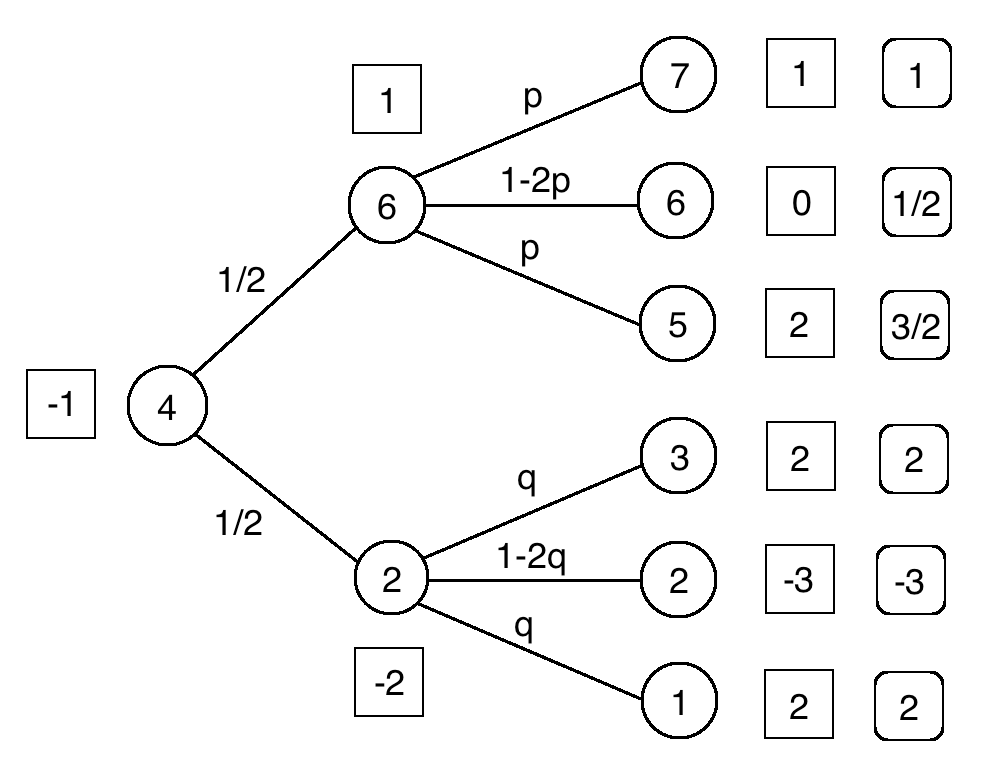}\ \ \ \ \ \ \ \ \ \ 
\includegraphics[scale=0.2]{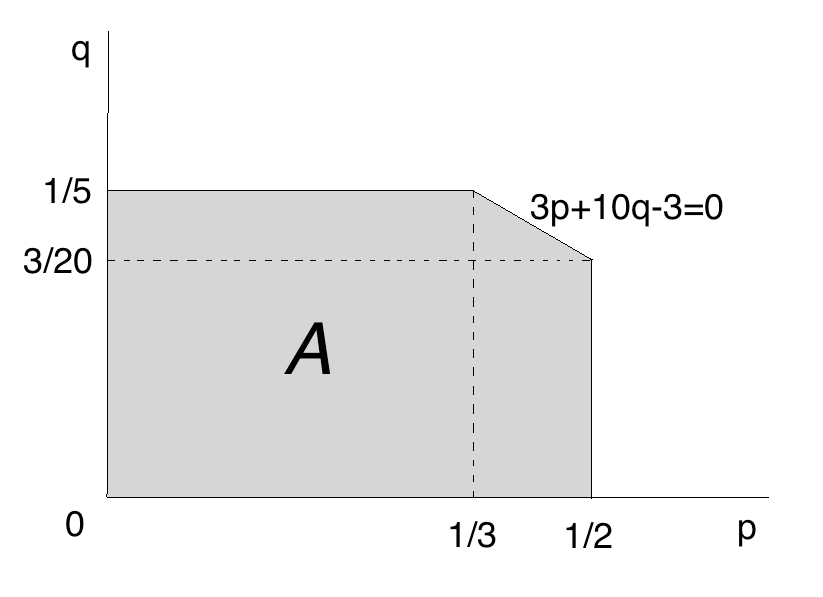}\\
\scriptsize{\textbf{Graph 1}\ \ \ \ \ \ \ \ \ \ \ \ \ \ \ \ \ \ \ \ \ \ \ \ \ \ \ \ \ \ \ \ \ \ \ \ \ \ \ \ \ \ \ \ \ \ \ \ \ \ \ \ \ \textbf{Graph 2}}
\end{center}
\vspace{0.2cm}

Consider a simple model given by Graph 1 above. The stock prices $S=(S_t)_{t=0,1,2}$, payoffs of the American option $h=(h_t)_{t=0,1,2}$, and payoffs of the European option $\psi$ are indicated by the numbers in the circles, squares with straight corners, and squares with rounded corners, respectively. Let $(\Omega,\mathcal{B}(\Omega))$ be the path space indicated by Graph 1, and let $(\mathcal{F}_t)_{t=0,1,2}$ be the filtration generated by $S$. Let $\P$ be a probability measure that is supported on $\Omega$. Hence any EMM would be characterized by the pair $(p,q)$ shown in Graph 1 with $0<p,q<1/2$. 

We assume that the American option $h$ can only be bought at time $t=0$ with price $\bar h=0$. Then in order to avoid arbitrage involving stock $S$ and American option $h$, we expect that the set
$$\cQ:=\left\{\Q\ \text{is an EMM}:\ \sup_{\tau\in\cT}\E_\Q h_\tau\leq 0\right\}$$
is not empty, where $\cT$ represents the set of stopping times. Equivalently, to avoid arbitrage, the set
$$A:=\left\{(p,q)\in\left(0,\frac{1}{2}\right)\times\left(0,\frac{1}{2}\right):\ \left(\frac{1}{2}[(3p)\vee 1]+\frac{1}{2}[(10q-3)\vee(-2)]\right)\vee(-1)\leq 0\right\}$$
should be nonempty. In Graph 2 above $A$ is indicated by the shaded area, which shows that $A\neq\emptyset$. 

Now consider the super-hedging price $\bar\pi(\psi)$ of the European option $\psi$ using semi-static trading strategies. That is,
$$\bar\pi(\psi):=\inf\{x:\ \exists(H,c,\tau)\in\mathcal{H}\times\R_+\times\cT,\text{ s.t. }x+H\cdot S+ch_\tau\geq\psi,\ \P-\text{a.s.}\},$$
where $\mathcal{H}$ is the set of adapted processes, and $H\cdot S=\sum_{t=0}^1H_t(S_{t+1}-S_t)$. One may expect that the super-hedging duality would be given by
$$\bar\pi(\psi)=\sup_{\Q\in\mathcal{Q}}\E_\Q\psi.$$
By calculation,
$$\sup_{\Q\in\mathcal{Q}}\E_\Q\psi=\sup_{(p,q)\in A}\left(\frac{3}{4}p+5q-\frac{5}{4}\right)=\left(\frac{3}{4}p+5q-\frac{5}{4}\right)\bigg|_{(\frac{1}{3},\frac{1}{5})}=0.$$
On the other hand, it can be shown that
\begin{eqnarray}
\notag \bar\pi(\psi)&=&\inf_{\tau\in\cT}\inf_{c\in\R_+}\inf\{x:\ \exists H\in\cH, \text{ s.t. } x+H\cdot S\geq\psi-ch_\tau\}\\
\notag &=&\inf_{\tau\in\cT}\inf_{c\in\R_+}\sup_{\Q\in\mathcal{M}}\E_\Q[\psi-ch_\tau]\\
\notag &=&\frac{1}{8},
\end{eqnarray}
where $\cM$ is the set of EMMs. Here we use the classical result of super-hedging for the second line, and the value in the third line can be calculated by brute force since we only have five stopping times.\footnote{For example, when
\begin{displaymath}
   \tau = \left\{
     \begin{array}{lr}
       2, & S_1=6,\\
       1, & S_1=2,
     \end{array}
   \right.
\end{displaymath} 
then
$$\inf_{c\in\R_+}\sup_{\Q\in\cM}\E_\Q[\psi-ch_\tau]=\inf_{c\geq 0}\sup_{0<p,q<\frac{1}{2}}\left[\left(\frac{3}{4}-\frac{3}{2}c\right)p+5q-\frac{5}{4}+c\right]=\frac{13}{8}$$
}
\textbf{Therefore, the super hedging price is strictly bigger than the sup over the EMMs $\Q\in\mathcal{Q}$}, i.e.,
$$\bar\pi(\psi)>\sup_{\Q\in\mathcal{Q}}\E_\Q\psi.$$
As a consequence, if we add $\psi$ into the market, and assume that we can only sell $\psi$ at $t=0$ with price $\underline\psi=1/16$ ($>0=\sup_{\Q\in\mathcal{Q}}\E_\Q\psi$), then the market would \textbf{admit no arbitrage, yet there is no $Q\in\cQ$, such that $\E_\Q[\psi]\geq\underline\psi$}. 

However, observe that $\psi=\frac{1}{2}(h_{\tau_{12}}+h_2)$, where
\begin{displaymath}
   \tau_{12} = \left\{
     \begin{array}{lr}
       1, & S_1=6,\\
       2, & S_1=2.
     \end{array}
   \right.
\end{displaymath}
This suggests that if we assume that $h$ is infinitely divisible, i.e., we can break one unit of $h$ into pieces, and exercise each piece separately, then we can show that {the super-hedging price of $\psi$ is $\sup_{\Q\in\mathcal{Q}}\E_\Q\psi=0$}. Now if we add $\psi$ into the market with selling price $\underline\psi<0$, then we can find $\Q\in\cQ$, such that $\E_\Q\psi>\underline\psi$.

\section{Setup and main results without model ambiguity}\label{s3}
In this section, we first describe the setup of our financial model (without model uncertainty). In particular, as suggested by the example in the last section, we shall assume that the American options are divisible. Then we shall provide the main results, including \thref{t1} for FTAP, \thref{t2} for sub-hedging, and \thref{t3} for utility maximization.

\subsection{Setup}
Let $(\Omega,\cF,(\cF_t)_{t=0,1,\dotso,T},\P)$ be a filtered probability space, where $\cF$ is assumed to be separable, and $T\in\mathbb{N}$ represents the time horizon in discrete time. Let $S=(S_t)_{t=0,\dotso,T}$ be an adapted process taking values in $\R^d$ which represents the stock price process. Let $f^i,g^j:\ \Omega\mapsto\R$ be $\cF_T$-measurable, representing the payoffs of European options, $i=1,\dotso,L$ and $j=1,\dotso,M$. We assume that we can buy \textit{and} sell each $f^i$ at time $t=0$ at price $\bar f^i$, and we can only buy but \textit{not} sell each $g^j$ at time $t=0$ with price $\bar g^j$.  Let $h^k=(h^k_t)_{t=0,\dotso,T}$ be an adapted process, representing the payoff process of an American option, $k=1,\dotso,N$. We assume that we can only buy but \textit{not} sell each $h^k$ at time $t=0$ with price $\bar h^k$. Denote $f=(f^1,\dotso,f^L)$ and $\bar f=(\bar f^1,\dotso, \bar f^L)$, and similarly for $g,\bar g, h$ and $\bar h$. For simplicity, we assume that $g$ and $h$ are bounded.

\begin{remark}
Here $g$ may represent the European options whose trade is quoted with bid-ask spreads. This is without loss of generality, since for any European option $\mathfrak{g}$ with bid price $\underline{\mathfrak{g}}$ and ask price $\overline{\mathfrak{g}}$, we can treat the option as two European options $\mathfrak{g}^1=-\mathfrak{g}$ and $\mathfrak{g}^2=\mathfrak{g}$ which can only be bought at price $-\underline{\mathfrak{g}}$ and $\overline{\mathfrak{g}}$ respectively. As for the American options $h$, we restrict ourself to only buy $h$. This is because if we sell American options, we will confront the risk of not knowing when the American options will be exercised. Moreover, if American options are sold, we need to consider non-anticipating trading strategies, and the problems will become much more complicated (see e.g., \cite{BHZ,ZZ7,ZZ6}). 
\end{remark}

\begin{definition}
An adapted process $\eta=(\eta_t)_{t=0,\dotso,T}$ is said to be a liquidating strategy, if $\eta_t\geq 0$ for $t=0,\dotso,T$, and
\begin{equation}\label{e75}
\sum_{t=1}^T{\eta_t}=1,\quad\P\text{-a.s.}.
\end{equation}
Denote $\T$ as the set of all liquidating strategies.
\end{definition}
\begin{remark}\label{r1}
Let us also mention the related concept of a randomized stopping time, which is a random variable $\gamma$ on the enlarged probability space $(\Omega\times[0,1],\cF\otimes\mathcal{B},\P\times\lambda)$, such that $\{\gamma=t\}\in\cF_t\otimes\mathcal{B}$ for $t=0,\dotso,T$, where $\mathcal{B}$ is the Borel sigma algebra on $[0,1]$ and $\lambda$ is the Lebesgue measure.  Let $\T'$ be the set of randomized stopping times. For $\gamma\in\T'$, its $\omega$-distribution $\xi=(\xi_t)_{t=0,\dotso,T}$ defined by
$$\xi_t(\cdot)=\lambda\{v:\ \gamma(\cdot,v)=t\},\quad t=0,\dotso,T,$$
is a member in $\T$. There is one-to-one correspondence between $\T$ and $\T'$ (up to a increasing rearrangement). We refer to \cite{Edgar} for these facts.

In spite of the one-to-one correspondence, the paths of a liquidating strategy and a randomized stopping time are quite different. A randomized stopping time is the strategy of flipping a coin to decide which stopping time to use (so the whole unit is liquidated only once), while a liquidating strategy is an exercising flow (so different parts of the whole unit are liquidated at different times).

Because of this difference, \thref{t1} (FTAP), \thref{t2} (hedging duality) and \thref{t3} (utility maximization duality) will not hold if we replace liquidating strategies with randomized stopping times. (For randomized stopping times, one may still consider FTAP and hedging on the enlarged probability space, and the results would be different.) For instance, in the example from last section, unlike liquidating strategies, we cannot merely use $h$ to super-hedge $\psi$ (on the enlarged probability space) via any randomized stopping time.  See \reref{r2} for more explanation for the case of utility maximization.     
\end{remark}

For each $\eta\in\T$ and American option $h^k$, denote $\eta(h^k)$ as the payoff of $h^k$ by using the liquidating strategy $\eta$. That is,
$$\eta(h^k)=\sum_{t=0}^T h^k_t \eta_t.$$
For $\mu=(\mu^1,\dotso,\mu^N)\in\T^N$, denote
$$\mu(h)=(\mu^1(h^1),\dotso,\mu^N(h^N)).$$
Let $\mathcal{H}$ be the set of adapted processes which represents the dynamical trading strategies for stocks. Let $(H\cdot S)_t:=\sum_{t=0}^{T-1}H_t(S_{t+1}-S_t)$, and denote $H\cdot S$ for $(H\cdot S)_T$ for short. For a semi-static trading strategy $(H,a,b,c,\mu)\in\mathcal{H}\times\R^L\times\R_+^M\times\R_+^N\times\T^N$, the terminal value of the portfolio starting from initial wealth 0 is given by
$$\PPPhi:=H\cdot S+a(f-\bar f)+b(g-\bar g)+c(\mu(h)-\bar h),$$
where $f-\bar f:=(f^1-\bar f^1,\dotso,f^L-\bar f^L)$, and $af$ represents the inner product of $a$ and $f$, and similarly for the other related terms. For $(H,a)\in\mathcal{H}\times\R^L$ we shall also use the notation
$$\PPhi:=H\cdot S+a(f-\bar f)$$
for short. From now on, when we write out the quintuple such as $(H,a,b,c,\mu)$, they are by default in $\mathcal{H}\times\R^L\times\R_+^M\times\R_+^N\times\T^N$ unless we specifically point out, and similarly for $(H,a)$. 

\subsection{Fundamental theorem of asset pricing}
\begin{definition}
We say no arbitrage (NA) holds w.r.t. $\bar g$ and $\bar h$, if for any $(H,a,b,c,\mu)$,
$$\PPPhi\geq 0\quad\P\text{-a.s.}\implies\PPPhi= 0\quad\P\text{-a.s.}.$$
We say strict no arbitrage (SNA) holds, if there exists $\eps_g\in(0,\infty)^M$ and $\eps_h\in(0,\infty)^N$ (from now on we shall use $\eps_g,\eps_h>0$ for short), such that NA holds w.r.t. $\bar g-\eps_g$ and $\bar h-\eps_h$.
\end{definition}

Define
$$\cQ:=\{\Q\ \text{is an EMM}:\ \E_\Q f=\bar f,\ \E_\Q g<\bar g,\ \sup_{\tau\in\mathcal{T}}\E_\Q h_\tau<\bar h\},$$
where $\cT$ is the set of stopping times, $\sup_{\tau\in\mathcal{T}}\E_\Q h_\tau:=(\sup_{\tau\in\mathcal{T}}\E_\Q h^1_\tau,\dotso,\sup_{\tau\in\mathcal{T}}\E_\Q h^N_\tau)$, and the expectation and equality/inequality above are understood in a component-wise sense.

\begin{theorem}[FTAP]\label{t1}
SNA $\EQUIV\cQ\neq\emptyset$.
\end{theorem}

\subsection{Sub-hedging} Let $\psi:\ \Omega\mapsto\R$ be $\cF_T$-measurable, which represents the payoff of a European option. Let $\phi=(\phi_t)_{t=0,\dotso,T}$ be an adapted process, representing the payoff process of an American option. For simplicity, we assume that $\psi$ and $\phi$ are bounded. Define the sub-hedging price of $\psi$
$$\pi_{eu}(\psi):=\sup\{x:\ \exists(H,a,b,c,\mu),\ \text{s.t.}\ \PPPhi+\psi\geq x\},$$
and the sub-hedging price of $\phi$
$$\pi_{am}(\phi):=\sup\{x:\ \exists(H,a,b,c,\mu)\ \text{and}\ \eta\in\T,\ \text{s.t.}\ \PPPhi+\eta(\phi)\geq x\}.$$

\begin{theorem}[Sub-hedging]\label{t2}
Let SNA hold. Then
\begin{equation}\label{e1}
\pi_{eu}(\psi)=\inf_{\Q\in\cQ}\E_\Q\psi,
\end{equation}
and
\begin{equation}\label{e2}
\pi_{am}(\phi)=\inf_{\Q\in\cQ}\sup_{\tau\in\cT}\E_\Q\phi_\tau.
\end{equation}
Moreover, there exists $(H^*,a^*,b^*,c^*,\mu^*)$ such that 
$$\Phi_{\bar g,\bar h}(H^*,a^*,b^*,c^*,\mu^*)+\psi\geq\pi_{eu}(\psi),$$
and there exists $(H^{**},a^{**},b^{**},c^{**},\mu^{**})$ and $\eta^{**}\in\T$ such that 
\begin{equation}\label{e4}
\Phi_{\bar g,\bar h}(H^{**},a^{**},b^{**},c^{**},\mu^{**})+\eta^{**}(\phi)\geq\pi_{am}(\phi).
\end{equation}
\end{theorem}
\begin{remark}
In fact, from the proof of \thref{t2} we have that
$$\pi_{am}(\phi)=\sup_{\eta\in\T}\inf_{\Q\in\cQ}\E_\Q[\eta(\phi)]=\inf_{\Q\in\cQ}\sup_{\eta\in\T}\E_\Q[\eta(\phi)]=\inf_{\Q\in\cQ}\sup_{\tau\in\cT}\E_\Q\phi_\tau.$$
However, the order of ``sup'' and ``inf'' in the duality \eqref{e2} cannot be exchanged. That is, it is possible that
$$\inf_{\Q\in\cQ}\sup_{\tau\in\cT}\E_\Q\phi_\tau>\sup_{\tau\in\cT}\inf_{\Q\in\cQ}\E_\Q\phi_\tau.$$
We refer to \cite[Example 2.1]{BHZ} for such an example. In fact, the right-hand-side above is the sub-hedging price of $\phi$ (even with model uncertainty) if $\phi$ is not divisible and only $S$ and $f$ are used for hedging. We refer to \cite[Theorem 2.1]{BHZ} regarding this point.
\end{remark}

\subsection{Utility maximization} Let $U: (0,\infty)\mapsto\R$ be a utility function, which is strictly increasing, strictly concave, continuously differentiable, and satisfies the Inada condition
$$\lim_{x\rightarrow 0+}U'(x)=\infty\quad\text{and}\quad\lim_{x\rightarrow\infty}U'(x)=0.$$
Consider the utility maximization problem
$$u(x):=\sup_{(H,a,b,c,\mu)\in\mathcal{A}(x)}\E_\P[U(\PPPhi)],\quad x>0,$$
where
$$\mathcal{A}(x):=\{(H,a,b,c,\mu):\ x+\PPPhi>0,\ \P\text{-a.s.}\},\quad x>0.$$
\begin{remark}
\cite{Hobson} also studies the utility maximization problem involving the liquidation of a \textit{given} amount of infinitely divisible American options. Unlike the problem in \cite{Hobson}, here we also incorporate the stocks and European options, and we need to decide how many shares of American options we need to buy at time $t=0$. Another difference is that \cite{Hobson} focuses on the primary problem of the utility maximization, while we shall mainly find the duality of the value function $u$.
\end{remark}

Let us define
$$V(y):=\sup_{x>0}[U(x)-xy],\quad y>0,$$
$$I:=-V'=(U')^{-1},$$
and for $x,y>0$,
\begin{eqnarray}
\notag \cX(x)&:=&\{X\ \text{adapted}:\ X_0=x,\ X_T=x+\PPPhi\geq 0\ \text{for some}\ (H,a,b,c,\mu)\},\\
\notag \cY(y)&:=&\{Y\geq 0\ \text{adapted}:\ Y_0=y,\ ((1+(H\cdot S)_t)Y_t)_{t=0,\dotso,T}\text{ is a $\P$-super-martingle}\\
\notag&&\text{for any}\ H\in\cH\ \text{satisfying }1+H\cdot S\geq 0,\ \E_\P X_TY_T\leq xy\ \text{for any}\ X\in\mathcal{X}(x)\}\\
\label{e5} \cC(x)&:=&\{p\in\bL_+^0:\ p\leq X_T\text{ for some }X\in\cX(x)\},\\
\label{e6} \cD(y)&:=&\{q\in\bL_+^0:\ q\leq Y_T\text{ for some }Y\in\cY(y)\},
\end{eqnarray}
where $\bL_+^0$ is the set of random variables that are nonnegative $\P$-a.s.. Then we have that
$$u(x)=\sup_{p\in\cC(x)}\E_\P[U(p)],\quad x>0.$$
Let us also define
$$v(y):=\inf_{q\in\cD(y)}\E_\P[V(q)],\quad y>0.$$

Below is the main result of utility maximization.
\begin{theorem}[Utility maximization duality]\label{t3}
Let SNA hold. Then we have the following.
\begin{itemize}
\item[i)] $u(x)<\infty$ for any $x>0$, and there exists $y_0>0$ such that $v(y)<\infty$ for any $y>y_0$. Moreover, $u$ and $v$ are conjugate:
$$v(y)=\sup_{x>0}[u(x)-xy],\quad y>0\quad\text{and}\quad u(x)=\inf_{y>0}[v(y)+xy],\quad x>0.$$
Furthermore, $u$ is continuous differentiable on $(0,\infty)$, $v$ is strictly convex on $\{v<\infty\}$, and
$$\lim_{x\rightarrow 0+}u'(x)=\infty\quad\text{and}\quad\lim_{y\rightarrow\infty}v'(y)=0.$$
\item[ii)] If $v(y)<\infty$, then there exists a unique $\hat q(y)\in\cD(y)$ that is optimal for $v(y)$.
\item[iii)] If $U$ has asymptotic elasticity strictly less than 1, i.e.,
$$\text{AE}(U):=\limsup_{x\rightarrow\infty}\frac{xU'(x)}{U(x)}<1,$$
Then we have the following.
\begin{itemize}
\item [a)] $v(y)<\infty$ for any $y>0$, and $v$ is continuously differentiable on $(0,\infty)$. $u'$ and $v'$ are strictly decreasing, and satisfy
$$\lim_{x\rightarrow\infty}u'(x)=0\quad\text{and}\quad\lim_{y\rightarrow 0+}v'(y)=-\infty.$$
Besides, $|\text{AE}(u)|\leq |\text{AE}(U)|<1$.
\item[b)] There exists a unique $\hat p(x)\in\cC(x)$ that is optimal for $u(x)$. If $\hat q(y)\in\cD(y)$ is optimal for $v(y)$, where $y=u'(x)$, then
$$\hat p(x)=I(\hat q(y)),$$
and
$$\E_\P[\hat p(x)\hat q(y)]=xy.$$
\item[c)] We have that
$$u'(x)=\E_\P\left[\frac{\hat p(x)U'(\hat p(x))}{x}\right]\quad\text{and}\quad v'(y)=\E_\P\left[\frac{\hat q(y)V'(\hat q(y))}{y}\right].$$
\end{itemize}
\end{itemize}
\end{theorem}

\begin{remark}\label{r2}
We cannot replace the liquidating strategies with randomized stopping times since the two types of strategies yield to very different optimization problems:
\begin{eqnarray}
\notag \E_\P U(\eta(\phi))&=&\E_\P\left[U\left(\sum_{t=0}^T\phi_t\eta_t\right)\right],\quad\text{if $\eta$ is a liquidating strategy},\\
\notag \E_{\P\times\lambda} U(\phi_\gamma)&=&\E_\P\left[\sum_{t=0}^T U\left(\phi_t\right)\eta_t\right],\quad\text{if $\eta$ is the $\omega$-distribution of $\gamma\in\T'$}.
\end{eqnarray}
\end{remark}

\section{Proof of Theorems \ref{t1}-\ref{t3}}
\begin{proof}[\textbf{Proof of \thref{t1}}]
``$\Longleftarrow$'': Let $\Q\in\cQ$. Then there exists $\eps_g,\eps_h>0$, such that
$$\E_\Q g<\bar g-\eps_g\quad{and}\quad\sup_{\tau\in\cT}\E_\Q h_\tau<\bar h-\eps_h.$$
Thanks to the one-to-one correspondence between $\T$ and $\T'$, we have that for any $\Q\in\cQ$, 
$$\sup_{\eta\in\T}\E_\Q[\eta(h^i)]=\sup_{\tau\in\cT}\E_\Q h^i_\tau,\quad i=1,\dotso,N,$$
see e.g., \cite[Proposition 1.5]{Edgar}. Then it is easy to see that NA w.r.t. $\bar g-\eps_g,\bar h-\eps_h$ holds, and thus SNA holds.

``$\Longrightarrow$'': We shall proceed in three steps.

\textbf{Step 1}. Define
$$\cI:=\{\PPhi-W:\ \text{for some }(H,a)\text{ and }W\in\bL_+^0\}\cap\bL^\infty,$$
where $\bL^\infty$ is the set of bounded random variables. We shall show that $\cI$ is sequentially closed under weak star topology in this step.

Let $(X^n)_{n=1}^\infty\subset\cI$ such that
$$X^n=\Phi(H^n,a^n)-W^n\stackrel{w^*}\longrightarrow X\in\bL^\infty,$$
where the notation ``$\stackrel{w^*}\longrightarrow$'' represents the convergence under the weak star topology. Then there exist $(Y^m)_{m=1}^\infty$ which are convex combinations of $(X^n)_n$, such that $Y^m\rightarrow X$ a.s. (see e.g., the argument regarding passing from weak star convergence to almost sure convergence below Definition 3.1 on page 35 in \cite{Sch}). Since $\cI$ is convex, $(Y^m)_m\subset\cI$. By \cite[Theorem 2.2]{nutz2}, $\cI$ is closed under $\P$-a.s. convergence. This implies $X\in\cI$.

\textbf{Step 2}. By SNA, there exist $\eps_g,\eps_h>0$, such that NA  holds w.r.t. $\bar g-\eps_g$ and $\bar h-\eps_h$. Then NA also holds w.r.t. $\bar g-\eps_g/2$ and $\bar h-\eps_h/2$. Define
$$\cJ:=\left\{\Phi_{\bar g-\frac{1}{2}\eps_g,\bar h-\frac{1}{2}\eps_h}(H,a,b,c,\mu)-W:\ \text{for some }(H,a,b,c,\mu)\text{ and }W\in\bL_+^0\right\}\cap\bL^\infty.$$
We shall show that $\cJ$ is sequentially closed under weak star topology.

Let $(X^n)_{n=1}^\infty\subset\cJ$ such that
$$X^n=\Phi_{\bar g-\frac{1}{2}\eps_g,\bar h-\frac{1}{2}\eps_h}(H^n,a^n,b^n,c^n,\mu^n)-W^n\stackrel{w^*}\longrightarrow X\in\bL^\infty.$$
We consider the following two cases: 
$$\liminf_{n\rightarrow\infty}||(b^n,c^n)||<\infty\quad\text{and}\quad\liminf_{n\rightarrow\infty}||(b^n,c^n)||=\infty,$$
where $||\cdot||$ represents the sup norm.

\textbf{Case (i) $\liminf_{n\rightarrow\infty}||(b^n,c^n)||<\infty$}. Without loss of generality, assume that $(b^n,c^n)\rightarrow (b,c)\in\R^M\times\R^N$. As $\mathcal{F}$ is separable, $\mathbb{L}^1$ is also separable. Then by the sequential Banach-Alaoglu theorem, there exists $\mu=(\mu_0,\dotso,\mu_T)\in\T^N$, such that up to a subsequence $\mu^n_t\stackrel{w^*}\longrightarrow\mu_t$ for $t=0,\dotso,T$. Since $h$ is bounded,
$$\mu^n(h)\stackrel{w^*}\longrightarrow\mu(h).$$
Then we have that
$$b^n\left(g-\left(\bar g-\frac{1}{2}\eps_g\right)\right)+c^n\left(\mu^n(h)-\left(\bar h-\frac{1}{2}\eps_h\right)\right)\stackrel{w^*}\longrightarrow b\left(g-\left(\bar g-\frac{1}{2}\eps_g\right)\right)+c\left(\mu(h)-\left(\bar h-\frac{1}{2}\eps_h\right)\right).$$
Hence,
$$\Phi(H^n,a^n)-W^n\stackrel{w^*}\longrightarrow X-b\left(g-\left(\bar g-\frac{1}{2}\eps_g\right)\right)-c\left(\mu(h)-\left(\bar h-\frac{1}{2}\eps_h\right)\right)\in\bL^\infty.$$
Then by Step 1, there exists $(H,a)$ and $W\in\bL_+^0$ such that
$$\Phi(H,a)-W=X-b\left(g-\left(\bar g-\frac{1}{2}\eps_g\right)\right)-c\left(\mu(h)-\left(\bar h-\frac{1}{2}\eps_h\right)\right).$$
Therefore,
$$X=\Phi_{\bar g-\frac{1}{2}\eps_g,\bar h-\frac{1}{2}\eps_h}(H,a,b,c,\mu)-W\in\cJ.$$

\textbf{Case (ii) $\liminf_{n\rightarrow\infty}||(b^n,c^n)||=\infty$}. Without loss of generality, assume that $d^n:=||(b^n,c^n)||>0$ for any $n$. We have that
$$\frac{X^n}{d^n}=\Phi_{\bar g-\frac{1}{2}\eps_g,\bar h-\frac{1}{2}\eps_h}\left(\frac{H^n}{d^n},\frac{a^n}{d^n},\frac{b^n}{d^n},\frac{c^n}{d^n},\mu^n\right)-\frac{W^n}{d^n}\stackrel{w^*}\longrightarrow 0.$$
Then by Case (i), there exist $(H',a',b',c',\mu')$ and $W'\in\bL_+^0$, such that
$$\Phi_{\bar g-\frac{1}{2}\eps_g,\bar h-\frac{1}{2}\eps_h}(H',a',b',c',\mu')-W'=0.$$
Moreover, $b',c'\geq 0$ and at least one component of $(b',c')$ equals 1. Hence
$$\Phi_{\bar g-\eps_g,\bar h-\eps_h}(H',a',b',c',\mu')>0,\quad\P\text{-a.s.},$$
which contradicts NA w.r.t. $\bar g-\eps_g$ and $\bar h-\eps_h$.

\textbf{Step 3}. Since $\cJ$ is convex and sequentially closed under the weak star topology, it is weak star closed by \cite[Corollary 5.12.7]{Conway}. Moreover, because NA holds w.r.t. $\bar g-\eps_g/2$ and $\bar h-\eps_h/2$, $\cJ\cap L_+^\infty=\{0\}$. Then by the abstract version of Theorem 1.1 in \cite{Sch} as formulated below Remark 3.1 in the same paper, there exists $q\in\mathbb{L}^1$ such that $q$ is a.s. strictly positive, and $\E_\P[qX]\leq 0$ for any $X\in\cJ$. Now define the measure $\Q$ by the Radon-Nikodym derivative $d\Q/d\P:=q/\E_\P[q]$. Then it can be seen that $\Q$ is an EMM satisfying
$$\E_\Q f=\bar f,\quad\E_\Q g\leq\bar g-\eps_g,\quad\text{and}\quad\sup_{\tau\in\mathcal{T}}\E_\Q h_\tau\leq\bar h-\eps_h.$$
In particular, $\cQ\neq\emptyset$.
\end{proof}

\begin{proof}[\textbf{Proof of \thref{t2}}]
We shall only prove the results for $\phi$. The case for $\psi$ is similar, and in fact simpler. Let us first prove \eqref{e2}. It can be shown that
$$\pi_{am}(\phi)\leq\sup_{\eta\in\T}\inf_{\Q\in\cQ}\E_\Q[\eta(\phi)]\leq\inf_{\Q\in\cQ}\sup_{\eta\in\T}\E_\Q[\eta(\phi)]=\inf_{\Q\in\cQ}\sup_{\tau\in\cT}\E_\Q\phi_\tau.$$
If $\pi_{am}(\phi)<\inf_{\Q\in\cQ}\sup_{\tau\in\cT}\E_\Q\phi_\tau$, then take $\bar\phi\in\R$ such that
\begin{equation}\label{e3}
\pi_{am}(\phi)<\bar\phi<\inf_{\Q\in\cQ}\sup_{\tau\in\cT}\E_\Q\phi_\tau.
\end{equation}
Now we add $\phi$ into the market, and we assume that $\phi$ can only be bought at time $t=0$ with price $\bar\phi$. Then since $\bar\phi>\pi_{am}(\phi)$, SNA also holds when $\phi$ is involved (i.e., when the market consists of $S$ traded dynamically, and $f,g,h,\phi$ traded statically). As a consequence, there exists $\Q\in\cQ$ such that $\sup_{\tau\in\cT}\E_\Q\phi_\tau<\bar\phi$ by \thref{t1}, which contradicts \eqref{e3}. Therefore, we have that \eqref{e2} holds. Similarly we can show that \eqref{e1} holds.

Next, let us prove the existence of an optimal sub-hedging strategy for $\phi$. It can be shown that
\begin{eqnarray}
\notag \pi_{am}(\phi)&=&\sup_{b\in\R_+^M,c\in\R_+^N}\sup_{\mu\in\T^N,\eta\in\T}\sup\{x:\ \exists (H,a),\ \text{s.t.}\ \PPPhi+\eta(\phi)\geq x\}\\
\notag &=&\sup_{b\in\R_+^M,c\in\R_+^N}\sup_{\mu\in\T^N,\eta\in\T}\inf_{\Q\in\cQ_f}\E_\Q[b(g-\bar g)+c(\mu(h)-\bar h)+\eta(\phi)],
\end{eqnarray}
where
$$\cQ_f:=\{\Q\ \text{is an EMM}:\ \E_\Q f=\bar f\},$$
and we apply the superhedging Theorem on page 828 in \cite{nutz2} for the second line. We shall proceed in three steps to show the existence of $(H^{**},a^{**},b^{**},c^{**},\mu^{**})$ and $\eta^{**}$ for \eqref{e4}.

\textbf{Step 1}. Consider the map $F: \R_+^M\times\R_+^N\mapsto\R$,
$$F(b,c):=\sup_{\mu\in\T^N,\eta\in\T}\inf_{\Q\in\cQ_f}\E_\Q[b(g-\bar g)+c(\mu(h)-\bar h)+\eta(\phi)].$$
For $(b,c),(b',c')\in\R_+^M\times\R_+^N$,
\begin{eqnarray}
\notag&&|F(b,c)-F(b',c')|\\
\notag &&\leq\sup_{\mu\in\T^N,\eta\in\T}\left|\inf_{\Q\in\cQ_f}\E_\Q[b(g-\bar g)+c(\mu(h)-\bar h)+\eta(\phi)]-\inf_{\Q\in\cQ_f}\E_\Q[b'(g-\bar g)+c'(\mu(h)-\bar h)+\eta(\phi)]\right|\\
\notag &&\leq\sup_{\mu\in\T^N,\eta\in\T}\sup_{\Q\in\cQ_f}\left|\E_\Q[b(g-\bar g)+c(\mu(h)-\bar h)+\eta(\phi)]-\E_\Q[b'(g-\bar g)+c'(\mu(h)-\bar h)+\eta(\phi)]\right|\\
\notag &&\leq\sup_{\mu\in\T^N,\eta\in\T}\sup_{\Q\in\cQ_f}\E_\Q[|b-b'||g-\bar g|+|c-c'||\mu(h)-\bar h|]\\
\notag &&\leq K(M+N)||(b,c)-(b',c')||,
\end{eqnarray}
where $|b-b'|:=(|b^1-{b'}^1|,\dotso,|b^M-{b'}^M|)$ and similar for the other related terms, and $K>0$ is a constant such that 
$$||g(\cdot)-\bar g||,||h_t(\cdot)-\bar h||,||\phi_t(\cdot)||\leq K,\quad\forall(t,\omega)\in\{0,\dotso,T\}\times\Omega.$$
Hence $F$ is continuous.

\textbf{Step 2}. Now take $\Q\in\cQ\subset\cQ_f$. Let
$$\eps:=\min_{1\leq i\leq M}\left\{\bar g^i-\E_\Q g^i\right\}\wedge\min_{1\leq i\leq N}\left\{\bar h^i-\sup_{\tau\in\cT}\E_\Q h^i_\tau\right\}>0.$$
Then
$$\sup_{b\in\R_+^M,c\in\R_+^N}F(b,c)\geq F(0,0)\geq-K>-2K\geq\sup_{||(b,c)||>\frac{3K}{\eps}}F(b,c).$$
As a consequence,
$$\sup_{b\in\R_+^M,c\in\R_+^N}F(b,c)=\sup_{||(b,c)||\leq\frac{3K}{\eps}}F(b,c).$$
By the continuity of $F$ from Step 1, there exists $(b^{**},c^{**})\in\R_+^M\times\R_+^N$, such that
$$\pi_{am}(\phi)=\sup_{b\in\R_+^M,c\in\R_+^N}F(b,c)=F(b^{**},c^{**})=\sup_{\mu\in\T^N,\eta\in\T}\inf_{\Q\in\cQ_f}\E_\Q[b^{**}(g-\bar g)+c^{**}(\mu(h)-\bar h)+\eta(\phi)].$$

\textbf{Step 3}. For any $\Q\in\cQ_f$, the map
$$(\mu,\eta)\mapsto\E_\Q[b^{**}(g-\bar g)+c^{**}(\mu(h)-\bar h)+\eta(\phi)]=\E_\P\left[\frac{d\Q}{d\P}\left(b^{**}(g-\bar g)+c^{**}(\mu(h)-\bar h)+\eta(\phi)\right)\right]$$
is continuous under the Baxter-Chacon topology\footnote{The sequence $\{\eta^n=(\eta_0^n,\dotso,\eta_T^n):\ n\in\mathbb{N}\}\subset\T$ is said to converge to $\eta\in\T$ in the Baxter-Chacon topology, if for any $Y\in\mathbb{L}^1$,
$$\lim_{n\rightarrow\infty}\E_\P\left[Y\eta_t^n\right]=\E_\P\left[Y\eta_t^n\right],\quad t=0,\dotso,T.$$
That is, the Baxter-Chacon topology is induced by the weak-star topology. We refer this to e.g., \cite{Edgar}.}. Then the map
$$(\mu,\eta)\mapsto\inf_{\Q\in\cQ_f}\E_\Q[b^{**}(g-\bar g)+c^{**}(\mu(h)-\bar h)+\eta(\phi)]$$
is upper semi-continuous under the Baxter-Chacon topology. By \cite[Theorem 1.1]{Edgar}, the set $\T^N\times\T$ is compact under the Baxter-Chacon topology. Hence there exists $(\mu^{**},\eta^{**})\in\T^N\times\T$, such that
\begin{eqnarray}
\notag \pi_{am}(\phi)&=&\sup_{\mu\in\T^N,\eta\in\T}\inf_{\Q\in\cQ_f}\E_\Q[b^{**}(g-\bar g)+c^{**}(\mu(h)-\bar h)+\eta(\phi)]\\
\notag &=&\inf_{\Q\in\cQ_f}\E_\Q[b^{**}(g-\bar g)+c^{**}(\mu^{**}(h)-\bar h)+\eta^{**}(\phi)]\\
\notag &=&\sup\{x:\ \exists(H,a),\ \text{s.t.}\ \Phi_{\bar g,\bar h}(H,a,b^{**},c^{**},\mu^{**})+\eta^{**}(\phi)\geq x\},
\end{eqnarray}
where we apply the Superhedging Theorem in \cite{nutz2} for the third line. By the same theorem in \cite{nutz2}, there exists $(H^{**},a^{**})$ such that
$$\Phi_{\bar g,\bar h}(H^{**},a^{**},b^{**},c^{**},\mu^{**})+\eta^{**}(\phi)\geq\pi_{am}(\phi).$$
\end{proof}

\begin{proof}[\textbf{Proof of \thref{t3}}]
Recall $\cC(x)$ defined in \eqref{e5} and $\cD(x)$ defined in \eqref{e6}, and denote $\cC:=\cC(1)$ and $\cD:=\cD(1)$. By Theorems 3.1 and 3.2 in \cite{Sch1}, it suffices to show that $\cC$ and $\cD$ have the following properties:
\begin{itemize}
\item[1)] $\cC(1)$ and $cD(1)$ are convex, solid, and closed in the topology of convergence in measure.
\item[2)] For $p\in\bL_+^0$,
$$p\in\cC\EQUIV\E_\P[pq]\leq 1\quad\text{for}\ \forall q\in\cD.$$
For $q\in\bL_+^0$,
$$q\in\cD\EQUIV\E_\P[pq]\leq 1\quad\text{for}\ \forall p\in\cC.$$
\item[3)] $\cC$ is bounded in probability and contains the identity function $\mathds{1}$.
\end{itemize}
It is easy to see that $\cC$ and $\cD$ are convex and solid, $\E_\P[pq]\leq 1$ for any $p\in\cC$ and $q\in\cD$, and $\cC$ contains the function $\mathds{1}$. We shall prove the rest of the properties in three parts.

\textbf{Part 1}. We shall show $\cC$ is bounded in probability. Take $\Q\in\cQ$. Then $d\Q/d\P\in\cD$, and
$$\sup_{p\in\cC}\E_\P\left[\frac{d\Q}{d\P}p\right]=\sup_{p\in\cC}\E_\Q p\leq 1.$$
Therefore, we have that
\begin{eqnarray}
\notag \sup_{p\in\cC}\P(p>C)&=&\sup_{p\in\cC}\P\left(\frac{d\Q}{d\P}p>\frac{d\Q}{d\P}C\right)\\
\notag &=&\sup_{p\in\cC}\left[\P\left(\frac{d\Q}{d\P}p>\frac{d\Q}{d\P}C,\ \frac{d\Q}{d\P}\leq\frac{1}{\sqrt{C}}\right)+\P\left(\frac{d\Q}{d\P}p>\frac{d\Q}{d\P}C,\ \frac{d\Q}{d\P}>\frac{1}{\sqrt{C}}\right)\right]\\
\notag &\leq&\P\left(\frac{d\Q}{d\P}\leq\frac{1}{\sqrt{C}}\right)+\sup_{p\in\cC}\P\left(\frac{d\Q}{d\P}p>\sqrt{C}\right)\\
\notag &\leq&\P\left(\frac{d\Q}{d\P}\leq\frac{1}{\sqrt{C}}\right)+\frac{1}{\sqrt{C}}\\
\notag &\rightarrow&0,\quad C\rightarrow\infty.
\end{eqnarray}

\textbf{Part 2}. We shall show that for $p\in\bL_+^0$, if $\E_\P[pq]\leq 1$ for any $q\in\cD$, then $p\in\cC$, and as a consequence, $\cC$ is closed under the topology of convergence in measure. Take $p\in\bL_+^0$ satisfying $\E_\P[pq]\leq 1$ for any $q\in\cD$. It is easy to see that for any $\Q\in\cQ$, the process $(\frac{d\Q}{d\P}|_{\cF_t})_{t=0,\dotso,T}$ is in $\cY(1)$. Therefore,
$$\sup_{\Q\in\cQ}\E_\Q p=\sup_{\Q\in\cQ}\E_\P\left[\frac{d\Q}{d\P}p\right]\leq 1.$$
Thanks to \thref{t2}, there exists $(H,a,b,c,\mu)$ such that
$$1+\PPPhi\geq p,$$
which implies that $p\in\cC$.

Now let $(p^n)_{n=1}^\infty\subset\cC$ such that $p^n\stackrel{\P}\longrightarrow p$. Then without loss of generality, we assume that $p^n\rightarrow p$ a.s.. For any $q\in\cD$, we have that
$$\E_\P[pq]\leq\liminf_{n\rightarrow\infty}\E_\P[p^nq]\leq 1.$$
This implies $p\in\cC$.

\textbf{Part 3}. We shall show that for $q\in\bL_+^0$, if $\E_\P[pq]\leq 1$ for any $p\in\cC$, then $q\in\cD$, and as a consequence, $\cD$ is closed under the topology of convergence in measure. Take $q\in\bL_+^0$ satisfying $\E_\P[pq]\leq 1$ for any $p\in\cC$. Since
$$\cC\supset\{p'\in\bL_+^0:\ p'\leq 1+H\cdot S,\text{ for some }H\in\cH\},$$
by \cite[Proposition 3.1]{Sch1} there exists a nonnegative adapted process $Y'=(Y'_t)_{t=0,\dotso,T}$, such that $q\leq Y'_T$, and for any $H\in\cH$ with $1+H\cdot S\geq 0$, $((1+(H\cdot S)_t)Y'_t)_{t=0,\dotso,T}$ is a $\P$-super-martingale. Now define
\begin{displaymath}
   Y_t = \left\{
     \begin{array}{ll}
       Y'_t, & t=0,\dotso,T-1,\\
       q, & t=T.
     \end{array}
   \right.
\end{displaymath}
Then it can be shown that $Y=(Y_t)_{t=0,\dotso,T}\in\cY(1)$. Since $q=Y_T$, $q\in\cD$. Similar to the argument in Part 2, we can show that $\cD$ is closed under the topology of convergence in measure.
\end{proof}

\section{Arbitrage and hedging under model uncertainty}

In this section, we extend the FTAP and sub-hedging results to the case of non-dominated model uncertainty. The main difficulty for the proof lies in the lack of a dominating measure.  The main idea for the proof is to discretize the path space and also to apply the minimax theorem. Theorem \ref{t4} and \ref{t5} are the main results of the model uncertainty case.

The notation in this and the next sections will be independent of the previous sections, yet we will borrow some concepts from Section \ref{s3} when there is no confusion. We follow the set-up in \cite{nutz2}. Let $\Omega$ be a complete separable metric space and $T\in\mathbb{N}$ be the time horizon. Let $\Omega_t:=\Omega^t$ be the $t$-fold Cartesian product for $t=1,\dotso,T$ (with convention $\Omega_0$ is a singleton). We denote by $\mathcal{F}_t$ the universal completion of $\mathcal{B}(\Omega_t)$. For each $t\in\{0,\dotso,T-1\}$ and $\omega\in\Omega_t$, we are given a nonempty convex set $\mathcal{P}_t(\omega)\subset\mathfrak{P}(\Omega_1)$ of probability measures.  Here $\mathcal{P}_t$ represents the possible models for the $t$-th period, given state $\omega$ at time $t$. We assume that for each $t$, the graph of $\mathcal{P}_t$ is analytic, which ensures that $\mathcal{P}_t$ admits a universally measurable selector, i.e., a universally measurable kernel $P_t:\ \Omega_t\rightarrow \mathfrak{P}(\Omega_t)$ such that $P_t(\omega)\in\mathcal{P}_t(\omega)$ for all $\omega\in\Omega_t$. Let
\begin{equation}\label{prob}
\mathcal{P}:=\{P_0\otimes\dotso\otimes P_{T-1}:\ P_t(\cdot)\in\mathcal{P}_t(\cdot),\ t=0,\dotso,T-1\},
\end{equation}
where each $P_t$ is a universally measurable selector of $\mathcal{P}_t$, and
$$P_0\otimes\dotso\otimes P_{T-1}(A)=\int_{\Omega_1}\dotso\int_{\Omega_1} 1_A(\omega_1,\dotso,\omega_T)P_{T-1}(\omega_1,\dotso,\omega_{T-1};d\omega_T)\dotso P_0(d\omega_1),\ \ \ A\in\Omega_T.$$

The concepts $S,f,g,h,\bar f,\bar g,\bar h,\cT,\T,\cH,\PPPhi,\PPhi$ (and related notation) are definied similarly as those in Section \ref{s3}, except that here we require $S,h$ and $\eta\in\T$ to be $(\mathcal{B}(\Omega_t))_t$-adapted, $f$ and $g$ to be $\mathcal{B}(\Omega_T)$-measurable, $\tau\in\cT$ to be an $(\mathcal{B}(\Omega_t))_t$-stopping time, $H\in\cH$ to be $(\mathcal{F}_t)_t$-adapted, and the summation in \eqref{e75} holds for every $\omega\in\Omega_T$. We assume that $g$ is bounded from below.

\begin{remark}
In order to apply the results in \cite{ZZ3} (the results in \cite{ZZ3} is based on those in \cite{nutz2}), we require that $H\in\mathcal{H}$ to be $(\mathcal{F}_t)_t$-adapted, while $S,h$ and $\eta\in\T$ to be $(\mathcal{B}(\Omega_t))_t$-adapted, and $f$ and $g$ to be $\mathcal{B}(\Omega_T)$-measurable. In \cite{nutz2}, such different kinds of measurability is chosen for $\cH,S$ and $f$ because of the measurable selection argument.
\end{remark}

Recall the definition of (strict) no arbitrage in the quasi-surely sense (see e.g., \cite{ZZ3,nutz2}).

\begin{definition}
We say NA$(\cP)$ holds w.r.t. $\bar g$ and $\bar h$, if for any $(H,a,b,c,\mu)$,
$$\text{if}\quad\PPPhi\geq 0\ \cP\text{-q.s.}\footnote{We say a property holds $\cP$-q.s., if the property holds $P$-a.s. for any $P\in\cP$.},\quad\text{then}\quad\PPPhi=0\ \cP\text{-q.s.}.$$
We say SNA$(\cP)$ (w.r.t. $\bar g$ and $\bar h$), if there exists $\eps_g,\eps_h>0$, such that NA$(\cP)$ holds w.r.t. $\bar g-\eps_g$ and $\bar h-\eps_h$.
\end{definition}

Given $\mathcal{B}(\Omega_T)$-measurable European option $\psi$ and $\mathcal{B}(\Omega_t))_t$-adapted American option $\phi$, let us define the sub-hedging prices, 
$$\pi_{eu}(\psi):=\sup\{x:\ \exists(H,a,b,c,\mu),\ \text{s.t.}\ \PPPhi+\psi\geq x,\ \cP\text{-q.s.}\},$$
and
$$\pi_{am}(\phi):=\sup\{x:\ \exists(H,a,b,c,\mu)\ \text{and}\ \eta\in\T,\ \text{s.t.}\ \PPPhi+\eta(\phi)\geq x,\ \cP\text{-q.s.}\}.$$

For $\tilde g\in\R^M$ and $\tilde h\in\R^N$, define the set of martingale measures (MMs)
\begin{equation}\label{e92}
\cQ_{\tilde g,\tilde h}:=\{Q\lll\cP\footnote{We say $Q\lll\cP$, if $\exists P\in\cP$, such that $Q\ll P$.}:\ Q\text{ is an MM, }\E_Q f=\bar f,\ \E_Q g\leq\tilde g,\ \sup_{\tau\in\mathcal{T}}\E_Q h_\tau\leq\tilde h\}.
\end{equation}

We will make the following standing assumption.
\begin{assumption}\label{a1}
\
\begin{itemize}
\item [(1)] The set
$$\cQ_{\bar g}:=\{Q\lll\cP:\ Q\text{ is an MM, }\E_Q f=\bar f,\ \E_Q g\leq\bar g\}$$
is weakly compact.

\item[(2)] $g$ is bounded from below, and $\psi$ is bounded and continuous.

\item[(3)] For $k=1,\dotso,N$ and $t=1,\dotso,T$, $h_t^k$ and $\phi_t$ are bounded and uniformly continuous. 
\end{itemize}
\end{assumption}

\begin{remark}
The weak compactness of $\cQ_{\bar g}$ is used twice in the proof of the sub-hedging duality (see \eqref{e91} and \eqref{e901} in the next section). In particular, in \eqref{e91} to show the exchangeability of the sup over $\cQ_{\bar g}$ and the inf over $\T^{N+1}$, we use a discretization argument, in which the weak compactness guarantees the desired limiting property (see Step 3 in the proof of \leref{l1}). It can also be found in e.g. \cite{DS3,DS1}, that the weak compactness of related probability measure set plays a crucial role in the limiting argument. In \eqref{e901}, we directly apply the minimax theorem which requires the weak compactness of $\cQ_{\bar g}$.
\end{remark}

\begin{remark}
What we have in mind is a tree/lattice model in which the location of the nodes are unknown but are believed to reside in certain intervals. (This is a discrete-time version of the volatility uncertainty.) It is natural in fact to assume in these tree models that the stock price is bounded, and the set of martingale measures with ``bounded volatility'' is weakly compact. In general, $\cQ_{\bar g}$ is tight by the martingale property (when $S$ is the canonical process). Therefore, our assumption is actually just weak closedness.
\end{remark}

\begin{example}
If $\cP$ is the set of probability measures on a compact set $\tilde\Omega\subset\R^T$, $S$ is the canonical process on $\tilde\Omega$, $f$ is continuous, and $g$ is lower semi-continuous, then $\cQ_{\bar g}$ is weakly compact. In this case, we are in the flavor of the model independent setup, see e.g., \cite{Sch3}. That is, every scenario is possible within the path space $\tilde\Omega$. Note that the notion of arbitrage in this case is different from that in \cite{Sch3}.
\end{example}

\begin{example}
Similar to \cite{Sch3}, the existence of a European option with a super-linear grow condition (w.r.t. stock) may compactify $\mathcal{Q}_{\bar g}$. To be more specific, let $\Omega$ be a Banach space, and $\mathcal{P}$ be the set of probability measures on some closed subset of $\Omega_T$. Let $\kg: \Omega\mapsto\R$ be lower semi-continuous satisfying
\begin{equation}\label{ee1}
\liminf_{||\alpha||\rightarrow\infty}\frac{\kg(\alpha)}{||\alpha||}\in(0,\infty].
\end{equation}
We assume that one of the European options $g$ takes the form $\kg$, i.e., $g^1(\omega)=\kg(\omega_{t'})$ for some $t'\in\{1,\dotso,T\}$. Assume that for $l=1,\dotso,d$, $t=1,\dotso,T$, $i=1,\dotso,L$, and $j=2,\dotso,M$, $S_t^l$ and $f^i$ are continuous, and $g^j$ is lower semi-continuous, and
\begin{equation}\label{ee2}
\lim_{||\omega||\rightarrow\infty}\frac{S_t^l(\omega)}{\gamma(\omega)}=\lim_{||\omega||\rightarrow\infty}\frac{f^i(\omega)}{\gamma(\omega)}=\lim_{||\omega||\rightarrow\infty}\frac{g^j(\omega)^-}{\gamma(\omega)}=0,
\end{equation}
where  $\gamma(\omega):=\sum_{t=1}^T\kg(\omega_t)$, and ${g^j}^-$ is the negative part of $g^j$. Using an argument similar to the proof of Theorem 1.3 on page 9 in \cite{Sch3}, we can show that $\mathcal{Q}_{\bar g}$ is weakly compact. Indeed, \eqref{ee1} implies that $\mathcal{Q}_{\bar g}$ is tight and thus precompact, and \eqref{ee2} is some boundedness condition for $S,f,g$ when $||\omega||\rightarrow\infty$, which enables us to apply the properties of weak convergence to show the closedness of $\cQ_{\bar g}$.
\end{example}

Below are the main results for sub-hedging and FTAP under model uncertainty.
\begin{theorem}[Sub-hedging]\label{t4}
Assume SNA$(\cP)$ holds when only $S,f$ and $g$ are involved. Then under \asref{a1} we have that
\begin{equation}\notag
\pi_{eu}(\psi)=\inf_{Q\in\cQ_{\bar g,\bar h}}\E_Q \psi,
\end{equation}
and
\begin{equation}\label{e71}
\pi_{am}(\phi)=\inf_{Q\in\cQ_{\bar g,\bar h}}\sup_{\tau\in\cT}\E_Q \phi_\tau,
\end{equation}
where if $\cQ_{\bar g,\bar h}=\emptyset$, then $\pi_{eu}(\psi)=\pi_{am}(\phi)=\infty$.
Moreover, if $\cQ_{\bar g,\bar h}\neq\emptyset$, then there exist $Q_{eu},Q_{am}\in\cQ_{\bar g,\bar h}$ such that
$$\pi_{eu}(\psi)=\E_{Q_{eu}} \psi\quad\text{and}\quad\pi_{am}(\phi)=\sup_{\tau\in\cT}\E_{Q_{am}} \phi_\tau.$$
\end{theorem}

\begin{theorem}[FTAP]\label{t5}
Let \asref{a1} hold. Then SNA$(\cP)$ holds, if and only if there exist $\tilde g\in\R^M$ and $\tilde h\in\R^N$ with $\tilde g<\bar g$ and $\tilde h<\bar h$, such that for any $P\in\cP$, there exists $Q\in\cQ_{\tilde g,\tilde h}$ dominating $P$.
\end{theorem}

\begin{remark}\label{r5}
If SNA$(\cP)$ holds, then SNA$(\cP)$ also holds when only $S,f$ and $g$ are involved, and $\cQ_{\bar g,\bar h}\neq\emptyset$ by \thref{t5}.
\end{remark}

\section{Proof of Theorems \ref{t4} and \ref{t5}}

\begin{lemma}\label{l1}
Let $\cR$ be a convex, weakly compact set of probability measures on $(\Omega_T,\mathcal{B}(\Omega_T))$. Let \asref{a1}(3) hold. Then
\begin{equation}\label{e73}
\sup_{\substack{\mu^k\in\T\\k=1,\dotso,N}}\inf_{R\in\cR}\E_R\left[\sum_{k=1}^N\mu^k(h^k)\right]=\inf_{R\in\cR}\sup_{\substack{\mu^k\in\T\\k=1,\dotso,N}}\E_R\left[\sum_{k=1}^N\mu^k(h^k)\right]=\inf_{R\in\cR}\sup_{\substack{\tau^k\in\cT\\k=1,\dotso,N}}\E_R\left[\sum_{k=1}^N h_{\tau^k}^k\right].
\end{equation}
Moreover, the infimum in the third term above is attained.
\end{lemma}

\begin{proof}
For $\mu^k\in\T$, $\mu^k(h^k)$ may not be (semi-)continuous. Therefore, we cannot directly apply the minimax theorem for the first equality in \eqref{e73}. To overcome the difficulties coming from both the discontinuity of $\mu^k(h^k)$ and the non-dominancy of $\cR$, we will discretize $\Omega_T$ first and then take a limit.

To this end, for $n=1,2,\dotso$, let $(A_i^n)_{i\in\mathbb{N}}\subset\mathcal{B}(\Omega)$ be a countable partition of $\Omega$, such that the diameter of each $A_i^n$ is less than $1/n$. Take $\alpha_i^n\in\A_i^n$ for $i\in\mathbb{N}$, and define the map $\theta^n:\Omega\mapsto\Omega$,
$$\theta^n(\beta)=\alpha_j^n\quad\text{if}\quad\beta\in A_j^n\text{ for some }j.$$
Let $\xi^n:\ \Omega_T\mapsto\Omega_T$, such that each component of $\xi^n(\omega)$ is given by
$$(\xi^n(\omega))_t=\theta^n(\omega_t),\quad t=1,\dotso,T,\quad\omega=(\omega_1,\dotso,\omega_T)\in\Omega_T.$$
(Then $\xi^n$ can also be treated as an $(\mathcal{B}(\Omega_t))_t$-adapted process.) Let
$$\cR_n:=\{R\circ(\xi^n)^{-1}:\ R\in\cR\}.$$
We shall proceed in four steps to show \eqref{e73}.

\textbf{Step 1}. We show that
\begin{equation}\label{e76}
\limsup_{n\rightarrow\infty}\sup_{\substack{\mu^k\in\T\\k=1,\dotso,N}}\inf_{R\in\cR_n}\E_R\left[\sum_{k=1}^N\mu^k(h^k)\right]\leq\sup_{\substack{\mu^k\in\T\\k=1,\dotso,N}}\inf_{R\in\cR}\E_R\left[\sum_{k=1}^N\mu^k(h^k)\right].
\end{equation}
Fix $\eps>0$. Let $(\mu_n^1,\dotso,\mu_n^N)\in\T^N$ be such that
$$\inf_{R\in\cR_n}\E_R\left[\sum_{k=1}^N\mu_n^k(h^k)\right]\geq\sup_{\substack{\mu^k\in\T\\k=1,\dotso,N}}\inf_{R\in\cR_n}\E_R\left[\sum_{k=1}^N\mu^k(h^k)\right]-\eps.$$
Define $(\tilde\mu_n^1,\dotso,\tilde\mu_n^N)$ by $(\tilde\mu_n^k)_t=(\mu_n^k)_t\circ\xi^n$, for $t=0,\dotso,T$ and $k=1,\dotso,N$. Then it is easy to show that $(\tilde\mu_n^1,\dotso,\tilde\mu_n^N)\in\T^N$. For any $\tilde R\in\cR$, let $\tilde R_n:=\tilde R\circ (\xi^n)^{-1}\in\cR_n$. Then
$$E_{\tilde R_n}\left[\sum_{k=1}^N\mu_n^k(h^k)\right]=E_{\tilde R}\left[\sum_{k=1}^N\sum_{t=0}^T((\mu_n^k)_t\circ\xi^n)(h_t^k\circ\xi^n)\right]=E_{\tilde R}\left[\sum_{k=1}^N\sum_{t=0}^T(\tilde\mu_n^k)_t(h_t^k\circ\xi^n)\right].$$
Therefore,
$$\left|\E_{\tilde R_n}\left[\sum_{k=1}^N\mu_n^k(h^k)\right]-\E_{\tilde R}\left[\sum_{k=1}^N\tilde\mu_n^k(h^k)\right]\right|\leq\E_{\tilde R}\left[\sum_{k=1}^N\sum_{t=0}^T(\tilde\mu_n^k)_t\left|(h_t^k\circ\xi)-h_t^k\right|\right]\leq N\rho(1/n),$$
where $\rho$ is the modulus of continuity for $h$, i.e., for $t=1,\dotso,T$ and $k=1,\dotso,N$,
$$\left|h_t^k(\omega^1)-h_t^k(\omega^2)\right|\leq\rho\left(\max_{s=1,\dotso,t}|\omega_s^1-\omega_s^2|\right),\quad\omega^i=(\omega_1^i,\dotso,\omega_T^i)\in\Omega_T,\ i=1,2.$$
Hence, we have that
\begin{eqnarray}
\notag&&\sup_{\substack{\mu^k\in\T\\k=1,\dotso,N}}\inf_{R\in\cR_n}\E_R\left[\sum_{k=1}^N\mu^k(h^k)\right]-\eps\leq\inf_{R\in\cR_n}\E_R\left[\sum_{k=1}^N\mu_n^k(h^k)\right]\\
\notag&&\leq\E_{\tilde R_n}\left[\sum_{k=1}^N\mu_n^k(h^k)\right]\leq\E_{\tilde R}\left[\sum_{k=1}^N\tilde\mu_n^k(h^k)\right]+N\rho(1/n).
\end{eqnarray}
By the arbitrariness of $\tilde R$, we have that
$$\sup_{\substack{\mu^k\in\T\\k=1,\dotso,N}}\inf_{R\in\cR_n}\E_R\left[\sum_{k=1}^N\mu^k(h^k)\right]-\eps\leq\inf_{R\in\cR}\E_R\left[\sum_{k=1}^N\tilde\mu_n^k(h^k)\right]+N\rho(1/n).$$
Taking limsup on both sides above and then sending $\eps\searrow 0$, we have \eqref{e76} holds.

\textbf{Step 2}. We show that
$$\sup_{\substack{\mu^k\in\T\\k=1,\dotso,N}}\inf_{R\in\cR_n}\E_R\left[\sum_{k=1}^N\mu^k(h^k)\right]=\inf_{R\in\cR_n}\sup_{\substack{\mu^k\in\T\\k=1,\dotso,N}}\E_R\left[\sum_{k=1}^N\mu^k(h^k)\right].$$
As the domain of $(\xi^n)$ is countable, there exists a probability measure $R^*$ on the domain of $(\xi^n)$ that dominates $\cR_n$. Then we have that
\begin{eqnarray}
\notag&&\sup_{\substack{\mu^k\in\T\\k=1,\dotso,N}}\inf_{R\in\cR_n}\E_R\left[\sum_{k=1}^N\mu^k(h^k)\right]=\sup_{\substack{\mu^k\in\T\\k=1,\dotso,N}}\inf_{R\in\cR_n}\E_{R^*}\left[\frac{dR}{dR^*}\sum_{k=1}^N\mu^k(h^k)\right]\\
\notag&&=\inf_{R\in\cR_n}\sup_{\substack{\mu^k\in\T\\k=1,\dotso,N}}\E_{R^*}\left[\frac{dR}{dR^*}\sum_{k=1}^N\mu^k(h^k)\right]=\inf_{R\in\cR_n}\sup_{\substack{\mu^k\in\T\\k=1,\dotso,N}}\E_R\left[\sum_{k=1}^N\mu^k(h^k)\right],
\end{eqnarray}
where we apply the minimax theorem (see e.g., \cite[Corollary 2]{Frode}) for the second equality, and use the fact that $\T$ is compact and the map:
$$(\mu^1,\dotso,\mu^N)\mapsto\E_{R^*}\left[\frac{dR}{dR^*}\sum_{k=1}^N\mu^k(h^k)\right]$$
is continuous under the Baxter-Chacon topology (see e.g., \cite{Edgar}) w.r.t. $R^*$.

\textbf{Step 3}. We show that
\begin{equation}\label{e77}
\inf_{R\in\cR}\sup_{\substack{\tau^k\in\cT\\k=1,\dotso,N}}\E_R\left[\sum_{k=1}^N h_{\tau^k}^k\right]\leq\liminf_{n\rightarrow\infty}\inf_{R\in\cR_n}\sup_{\substack{\tau^k\in\cT\\k=1,\dotso,N}}\E_R\left[\sum_{k=1}^N h_{\tau^k}^k\right].
\end{equation}
By extracting a subsequence for the lower limit, we assume without loss of generality that the sequence $\left\{\inf_{R\in\cR_n}\sup_{\substack{\tau^k\in\cT\\k=1,\dotso,N}}\E_R\left[\sum_{k=1}^N h_{\tau^k}^k\right]\right\}$ converges. Fix $\eps>0$. Take $R_n\in\cR_n$ such that
$$\sup_{\substack{\tau^k\in\cT\\k=1,\dotso,N}}\E_{R_n}\left[\sum_{k=1}^N h_{\tau^k}^k\right]\leq\inf_{R\in\cR_n}\sup_{\substack{\tau^k\in\cT\\k=1,\dotso,N}}\E_R\left[\sum_{k=1}^N h_{\tau^k}^k\right]+\eps.$$
Let $\tilde R_n\in\cR$ be such that $R_n=\tilde R_n\circ(\xi^n)^{-1}$. As $\cR$ is weakly compact, there exists $\tilde R\in\cR$ such that up to a subsequence $\tilde R_n\xrightarrow{w}\tilde R$. Then for any bounded uniformly continuous function $\f\in\mathcal{B}(\Omega_T)$,
\begin{eqnarray}
\notag\left|\E_{R_n}\f-\E_{\tilde R}\f\right|&\leq&\left|\E_{R_n}\f-\E_{\tilde R_n}\f\right|+\left|\E_{\tilde R_n}\f-\E_{\tilde R}\f\right|\\
\notag&=&\left|\E_{\tilde R_n}(\f\circ\xi^n)-\E_{\tilde R_n}\f\right|+\left|\E_{\tilde R_n}\f-\E_{\tilde R}\f\right|\\
\notag&\leq&\E_{\tilde R_n}\left|(\f\circ\xi^n)-\f\right|+\left|\E_{\tilde R_n}\f-\E_{\tilde R}\f\right|\\
\notag&\leq&\rho_\f(1/n)+\left|\E_{\tilde R_n}\f-\E_{\tilde R}\f\right|\\
\notag&\rightarrow&0,\quad n\rightarrow\infty,
\end{eqnarray}
where $\rho_\f$ is the modulus of continuity of $\f$. Hence, $R_n\xrightarrow{w}\tilde R$. Since the map
$$R\mapsto\sup_{\tau^k\in\cT}\E_R\left[h_{\tau^k}^k\right]$$
is lower semi-continuous under weak topology (see e.g., \cite[Theorem 1.1]{Elton89}), the map
$$R\mapsto\sum_{k=1}^N\sup_{\tau^k\in\cT}\E_R\left[h_{\tau^k}^k\right]=\sup_{\substack{\tau^k\in\cT\\k=1,\dotso,N}}\E_R\left[\sum_{k=1}^N h_{\tau^k}^k\right]$$
is also lower semi-continuous. Therefore,
\begin{eqnarray}
\notag&&\lim_{n\rightarrow\infty}\inf_{R\in\cR_n}\sup_{\substack{\tau^k\in\cT\\k=1,\dotso,N}}\E_R\left[\sum_{k=1}^N h_{\tau^k}^k\right]+\eps\geq\liminf_{n\rightarrow\infty}\sup_{\substack{\tau^k\in\cT\\k=1,\dotso,N}}\E_{R_n}\left[\sum_{k=1}^N h_{\tau^k}^k\right]\\
\notag&&\geq\sup_{\substack{\tau^k\in\cT\\k=1,\dotso,N}}\E_{\tilde R}\left[\sum_{k=1}^N h_{\tau^k}^k\right]\geq\inf_{R\in\cR}\sup_{\substack{\tau^k\in\cT\\k=1,\dotso,N}}\E_R\left[\sum_{k=1}^N h_{\tau^k}^k\right].
\end{eqnarray}
Letting $\eps\searrow0$, we have \eqref{e77} holds.

\textbf{Step 4}. By steps 1-3, we have that
\begin{eqnarray}
\notag&&\sup_{\substack{\mu^k\in\T\\k=1,\dotso,N}}\inf_{R\in\cR}\E_R\left[\sum_{k=1}^N\mu^k(h^k)\right]\leq\inf_{R\in\cR}\sup_{\substack{\mu^k\in\T\\k=1,\dotso,N}}\E_R\left[\sum_{k=1}^N\mu^k(h^k)\right]\\
\notag&&=\inf_{R\in\cR}\sup_{\substack{\tau^k\in\cT\\k=1,\dotso,N}}\E_R\left[\sum_{k=1}^N h_{\tau^k}^k\right]\leq\liminf_{n\rightarrow\infty}\inf_{R\in\cR_n}\sup_{\substack{\tau^k\in\cT\\k=1,\dotso,N}}\E_R\left[\sum_{k=1}^N h_{\tau^k}^k\right]\\
\notag&&=\liminf_{n\rightarrow\infty}\inf_{R\in\cR_n}\sup_{\substack{\mu^k\in\T\\k=1,\dotso,N}}\E_R\left[\sum_{k=1}^N\mu^k(h^k)\right]=\liminf_{n\rightarrow\infty}\sup_{\substack{\mu^k\in\T\\k=1,\dotso,N}}\inf_{R\in\cR_n}\E_R\left[\sum_{k=1}^N\mu^k(h^k)\right]\\
\notag&&\leq\limsup_{n\rightarrow\infty}\sup_{\substack{\mu^k\in\T\\k=1,\dotso,N}}\inf_{R\in\cR_n}\E_R\left[\sum_{k=1}^N\mu^k(h^k)\right]\leq\sup_{\substack{\mu^k\in\T\\k=1,\dotso,N}}\inf_{R\in\cR}\E_R\left[\sum_{k=1}^N\mu^k(h^k)\right],
\end{eqnarray}
where the second and the fourth (in)equalities follows from \cite[Proposition 1.5]{Edgar}. Therefore, \eqref{e73} follows.

Finally, since the map
$$R\mapsto\sup_{\substack{\tau^k\in\cT\\k=1,\dotso,N}}\E_R\left[\sum_{k=1}^N h_{\tau^k}^k\right]$$ is lower semi-continuous, and $\cR$ is weakly compact, there exists $R^*\in\cR$ that attains the infimum of the third term in \eqref{e73}.
\end{proof}

\begin{proof}[\textbf{Proof of \thref{t4}}]
We will only prove the conclusions for American option $\phi$ (the case for European option $\psi$ is similar, and in fact slightly easier).

We have that
\begin{eqnarray}
\notag\pi_{am}(\phi)&=&\sup_{c\in\R_+^N}\sup_{\mu\in\T^N,\eta\in\T}\sup\{x\in\R:\ \exists(H,a,b), \text{ s.t. }\PPPhi+\eta(\phi)\geq x,\ \cP\text{-q.s.}\}\\
\notag&=&\sup_{c\in\R_+^N}\sup_{\mu\in\T^N,\eta\in\T}\inf_{Q\in\cQ_{\bar g}}\E_Q\left[c(\mu(h)-\bar h)+\eta(\phi)\right]\\
\label{e91}&=&\sup_{c\in\R_+^N}\inf_{Q\in\cQ_{\bar g}}\sup_{\substack{\tau\in\cT, \tau^k\in\cT\\k=1,\dotso,N}}\E_Q\left[\sum_{k=1}^N c^k(h_{\tau^k}^k-\bar h^k)+\phi_\tau\right],
\end{eqnarray}
where we apply \cite[Theorem 2.1(b)]{ZZ3} for the second equality, and \leref{l1} for the third equality. (Note that this is where we use the assumption that SNA$(\cP)$ holds when only $S,f$ and $g$ are involved.)

Now the map
$$c\mapsto\sup_{\substack{\tau\in\cT, \tau^k\in\cT\\k=1,\dotso,N}}\E_Q\left[\sum_{k=1}^N c^k(h_{\tau^k}^k-\bar h^k)+\phi_\tau\right],$$
is linear, and 
the map
$$Q\mapsto\sup_{\substack{\tau\in\cT, \tau^k\in\cT\\k=1,\dotso,N}}\E_Q\left[\sum_{k=1}^N c^k(h_{\tau^k}^k-\bar h^k)+\phi_\tau\right],$$
is lower semi-continuous (see step 3 in the proof of \leref{l1}) and convex. Thanks to the weak compactness of $\cQ_{\bar g}$, we can apply the minimax theorem (see e.g., \cite[Corollary 2]{Frode}) for \eqref{e91} and get that
\begin{eqnarray}
\label{e901} \pi_{am}(\phi)&=&\inf_{Q\in\cQ_{\bar g}}\sup_{c\in\R_+^N}\sup_{\substack{\tau\in\cT, \tau^k\in\cT\\k=1,\dotso,N}}\E_Q\left[\sum_{k=1}^N c^k(h_{\tau^k}^k-\bar h^k)+\phi_\tau\right]\\
\notag&=&\inf_{Q\in\cQ_{\bar g,\bar h}}\sup_{c\in\R_+^N}\sup_{\substack{\tau\in\cT, \tau^k\in\cT\\k=1,\dotso,N}}\E_Q\left[\sum_{k=1}^N c^k(h_{\tau^k}^k-\bar h^k)+\phi_\tau\right]\\
\notag&=&\inf_{Q\in\cQ_{\bar g,\bar h}}\sup_{\tau\in\cT}\E_Q[\phi_\tau].
\end{eqnarray}

Finally, we have that $\cQ_{\bar g,\bar h}$ is weakly compact, which implies the last statement of \thref{t4} by \leref{l1}. Indeed, for $(Q_n)_{n\in\mathbb{N}}\subset\cQ_{\bar g,\bar h}\subset\cQ_{\bar g}$, since $\cQ_{\bar g}$ is weakly compact by \asref{a1}(1), there exist $Q\in\cQ_{\bar g}$ and $(Q_{n_i})_{i\in\mathbb{N}}\subset(Q_n)_{n\in\mathbb{N}}$, such that $Q_{n_i}\xrightarrow{w} Q$. By \asref{a1}(3) and e.g., \cite[Theorem 1.1]{Elton89}, the map $R\mapsto\sup_{\tau\in\cT}\E_R\left[h_\tau^k\right]$ is lower semi-continuous for $k=1,\dotso,N$. Therefore,
$$\E_Q\left[h_\tau^k\right]\leq\liminf_{i\rightarrow\infty}\E_{Q_{n_i}}\left[h_\tau^k\right]\leq\bar h^k,\quad k=1,\dotso,N.$$
Hence, $Q\in\cQ_{\bar g,\bar h}$.
\end{proof}

\begin{proof}[\textbf{Proof of \thref{t5}}]
\textit{Sufficiency}. Assume that there exist $\tilde g<\bar g$ and $\tilde h<\bar h$, such that for any $P\in\cP$, there exists $Q\in\cQ_{\tilde g,\tilde h}$ dominating $P$. Then it is easy to show that NA$(\cP)$ holds w.r.t. $\tilde g$ and $\tilde h$, and thus SNA$(\cP)$ holds.

\textit{Necessity}. We will prove this by an induction on the number of liquid American options $N$. For $N=0$ the result follows from \cite[Theorem 2.1(a)]{ZZ3}. Now suppose the result holds for $N=n-1\in\mathbb{N}$. Let us consider $N=n$. For $k \in \{n-1,n\}$ denote NA$^k$, SNA$^k$, $\pi^k(\cdot)$ and $\cQ_{\cdot,\cdot}^k$ as the NA, SNA, subhedging price and martingale measure set defined in \eqref{e92} in terms of $S, f, g$ and $h^1,\dotso,h^k$, respectively. 

By SNA$^n(\cP)$, there exists $\hat h^n<\bar h^n$, such that  NA$^n(\cP)$ holds w.r.t. $\bar g$ and $(\bar h^1,\dotso,\bar h^{n-1},\hat h^n)$. It follows that
$$\pi^{n-1}(h^n)\leq\hat h^n,$$  
for otherwise, one would create an arbitrage by paying $\hat h^n$ to buy one unit of $h^n$ and getting $(\pi^{n-1}(h^n)+\hat h^n)/2$ via some trading strategy. As SNA$^n(\cP)$ holds, it can be seen that SNA$^{n-1}(\cP)$ also holds. Hence, by the induction hypothesis as well as \thref{t4} and \reref{r5}, we have that
\begin{equation}\label{e90}
\pi^{n-1}(h^n)=\inf_{Q\in\cQ_{\bar g,(\bar h^1,\dotso,\bar h^{n-1})}^{n-1}}\sup_{\tau\in\mathcal{T}}\E_Q[h_\tau^n]\leq\hat h^n<\bar h^n.
\end{equation}
Moreover, there exists $g_*\in\R^M$ and $h_*\in\R^{n-1}$ with $g_*<\bar g$ and $h_*<(\bar h^1,\dotso,\bar h^{n-1})$, such that for any $P\in\cP$, there exists $Q\in\cQ_{g_*, h_*}^{n-1}$ dominating $P$. 

By \asref{a1}(3), there exists $C>0$ such that $|h_t^n|<C$ for $t=0,\dotso,T$. Choose $\lambda\in(0,1)$ such that
$$\tilde h^n:=\lambda C+(1-\lambda)\frac{\hat h^n+\bar h^n}{2}<\bar h^n.$$
Now let
$$\tilde g:=\lambda g_*+(1-\lambda)\bar g,$$
and
$$\tilde h:=(\lambda h_*^1+(1-\lambda)\bar h^1,\dotso,\lambda h_*^{n-1}+(1-\lambda)\bar h^{n-1},\tilde h^n).$$
Let $P\in\cP$. We will show that there exists some $Q\in\cQ_{\tilde g,\tilde h}^n$ dominating $P$. Indeed, take $Q_*\in\cQ_{g_*,\tilde h_*}^{n-1}$ dominating $P$. By \eqref{e90}, there exists $\hat Q\in\cQ_{\bar g,(\bar h^1,\dotso,\bar h^{n-1})}^{n-1}$, such that
$$\sup_{\tau\in\mathcal{T}}\E_{\hat Q}[h_\tau^n]<\frac{\bar h^n+\hat h^n}{2}.$$
Let
$$Q_\lambda:=\lambda Q_*+(1-\lambda)\hat Q\gg P.$$
Obviously, $Q_\lambda\lll\cP$, $Q_\lambda$ is an MM, $\E_{Q_\lambda}f=\bar f$, $\E_{Q_\lambda}g\leq\tilde g$, and $\sup_{\tau\in\mathcal{T}}\E_{Q_\lambda}[h_\tau^k]\leq\tilde h^k$ for $k=1,\dotso,n-1$. Furthermore,
$$\sup_{\tau\in\cT}\E_{Q_\lambda}[h_\tau^n]=\sup_{\tau\in\cT}\left(\lambda\E_{Q_*}[h_\tau^n]+(1-\lambda)\E_{\hat Q}[h_\tau^n]\right)\leq\lambda C+(1-\lambda)\sup_{\tau\in\cT}\E_{\hat Q}[h_\tau^n]\leq\tilde h^n.$$
This implies $Q_\lambda\in\cQ_{\tilde g,\tilde h}^n$.
\end{proof}


\bibliographystyle{siam}
\bibliography{ref}
\end{document}